%% file: flan16B.tex
\documentclass[submission,copyright,creativecommons]{eptcs}
\providecommand{\event}{NCMA 2024}

\usepackage{iftex}
\usepackage[utf8]{inputenc}
\usepackage{amsmath}
\usepackage{amssymb} 
\usepackage{amsthm}
\usepackage{wrapfig}
\usepackage{xspace}
\usepackage{url}
\usepackage{tikz} 
\usepackage{tikz-qtree}
\usepackage{graphics} 
\usepackage{array}
\usepackage{multirow} 
\usepackage{float} 
\usepackage{faktor} 
\usepackage{leftidx}
\usepackage{enumerate}
\usepackage{booktabs}
\usepackage{comment}
\usepackage{shuffle}
\usepackage{hyperref}
\usepackage{cite}
\usepackage[capitalise]{cleveref}
\usepackage{tikzit}

\ifpdf
  \usepackage{underscore}  
  \usepackage[T1]{fontenc}
\fi

\input{TiKZit-General.tikzstyles}
\usetikzlibrary{cd,shapes,shapes.geometric,arrows,fit,calc,positioning,automata,decorations.pathmorphing,shadows} 
\tikzset{sstate/.style={state, inner sep=2pt, minimum size=7pt, circular drop shadow, fill=white}}
\tikzset{phantom/.style={state, white, inner sep=2pt, minimum size=7pt, fill=white}}

\author{
	Guilherme Duarte \qquad Nelma Moreira  \qquad Rogério Reis\thanks{
			This work was partially supported by CMUP, member of LASI, which is financed by national funds through FCT -- Funda\c c\~ao para a Ci\^encia e a Tecnologia, I.P., under the projects with reference UIDB/00144/2020 and UIDP/00144/2020.
		} \institute{
			CMUP \& DCC,
			Faculdade de Ci\^encias da Universidade do Porto\\
			Rua do Campo Alegre,
			4169-007 Porto,
			Portugal%
		}
		\email{\{guilherme.duarte,nelma.moreira,rogerio.reis\}@fc.up.pt}
		\and
		Luca Prigioniero
		\institute{
			Department of Computer Science,
			Loughborough University\\
			Epinal Way,
			Loughborough LE11 3TU,
			United Kingdom
		}
		\email{l.prigioniero@lboro.ac.uk}
}
\def\authorrunning{G.~Duarte, N.~Moreira, L.~Prigioniero \& R.~Reis}

\title{Operational State Complexity of Block Languages}
\def\titlerunning{Operational State Complexity of Block Languages}

\newtheorem{theorem}{Theorem}
\newtheorem{lemma}[theorem]{Lemma}
\newtheorem{example}{Example}
\newtheorem{corollary}{Corollary}
\newtheorem{proposition}{Proposition}

\newcommand{\Mod}[1]{\ (\mathrm{mod}\ #1)}

\input macros
\begin{document}
\maketitle

\begin{abstract}
In this paper we consider block languages, namely sets of words having the same length, and study the deterministic and nondeterministic state complexity of several operations on these languages. Being a subclass of finite languages, the upper bounds of operational state complexity known for finite languages apply for block languages as well. However, in several cases, smaller values were found. Block languages can be represented as bitmaps, which are a good tool to study their minimal finite automata and their operations, as we  illustrate here.
\end{abstract}

\section{Introduction}
In this paper we consider finite languages where all words have the same length, which are called \emph{homogeneous} or \emph{block} languages. Their investigation is mainly motivated by their applications to several contexts such as code theory~\cite{KMR:2018} and image processing~\cite{KarhOkho:2O14,KarhumakiK21}. We will focus on the state complexity of operations~\cite{yu94,gao13}. The \emph{deterministic} (\emph{nondeterministic}) \emph{state complexity} of a regular language~$L$ is the number of states of its minimal complete deterministic (nondeterministic, resp.) finite automaton.

Here, we are interested in operational complexity, that is the size of the model accepting a language resulting from an operation performed on one or more languages. In particular, the \emph{state complexity of an operation} (or \emph{operational state complexity}) on regular languages is the worst-case state complexity of a language resulting from the operation, considered as a function of the state complexities of the operands. As a subclass of finite languages, block languages inherit some properties known for that class, which  differ from   the existing ones for the class of regular languages~\cite{campeanu01}. Due to the fact that, in our case, all words have the same length, there are some gains in terms of state complexity. For example, it is known that the elimination of nondeterminism from an~$n$-state nondeterministic finite automaton for a block language costs~$2^{\Theta(\sqrt{n})}$ in size~\cite{KarhOkho:2O14}, which is smaller than the general case for finite languages.

A block language can be well characterized by its characteristic function which we denote by \emph{bitmap}. In particular, given an alphabet of size~$k$ and a length~$\ell$, a block language can be represented by a binary string of length~$k^\ell$, also called \emph{bitmap}, in which each symbol (or \emph{bit}) indicates whether the correspondent word, according to the lexicographical order, belongs to the language (bit equal to~$1$) or not (bit equal to~$0$). Duarte et al.~\cite{dmpr24a} used this representation as a tool to investigate several properties of block languages, namely how to convert bitmaps into minimal deterministic and nondeterministic finite automata and what are the maximal numbers of states that the resulting automata can have. In this paper, we also use bitmaps for studying the complexity of operations on block languages. Due to the distinguishing property of the length of the words, we study Boolean binary operations over block languages with the same length as well as block complement (i.e., $\Sigma^\ell\setminus L$). Nonetheless, we also consider operations such as concatenation, Kleene star, and Kleene plus, which are not closed for the class of block languages of a given length.

The paper is organized as follows. In the next section we fix notation and review the bitmap representation for block languages. In \cref{sec:sc}, we revise the operational state complexities of basic operations on finite languages. Then, we study the state complexity on block languages for the following operations: reversal, word addition and removal, intersection, union, block complement, concatenation, Kleene star, and plus. In \cref{tab:cblock}, we summarize our results and we conclude the paper in \cref{section:conclusions} by describing further lines of investigation.

\section{Preliminaries}\label{sec:preliminaries}
In this section we review some basic definitions about finite automata and languages and fix notation. Given an \emph{alphabet} $\Sigma$, a \emph{word}~$w$ is a sequence of symbols, and a \emph{language} $L\subseteq \Sigma^\star$ is a set of words on~$\Sigma$. The empty word is denoted by $\varepsilon$. The \emph{(left) quotient} of a language $L$ by a word~$w \in \Sigma^\star$ refers to the set $w^{-1}L = \{w' \in \Sigma^\star \mid ww' \in L\}$. The \emph{reversal} of a word~$w = \letter_0 \letter_1 \cdots \letter_{n-1}$ is denoted by~$\R{w}$ and is obtained by reversing the order of the symbols of~$w$, that is~$\R{w} = \letter_{n-1} \letter_{n-2} \cdots \letter_0$.  The reversal of a language $L$ is $\R{L}=\{\; \R{w}\mid w\in L\;\}$. Given two integers $i,j$ with $i<j$, let $[i,j]$ denote the set of integers from $i$ to $j$, including both $i$ and $j$, namely $\{i, \ldots, j\}$. Moreover, we shall omit the left bound if it is equal to $0$, thus $[j] = \{0,\ldots,j\}$.

A \emph{nondeterministic finite automaton} (\nfa) is a five-tuple $\aut{A}=\langle Q,\Sigma,\delta,I, F\rangle $ where $Q$ is a finite set of states, $\Sigma$ is a finite alphabet, $I\subseteq Q$ is the set of initial states, $F \subseteq Q$ is the set of final states, and~$\delta: Q \times \Sigma \to 2^Q$ is the transition function. We consider the \emph{size} of an \nfa as its number of states. The transition function can be extended to words and sets of states in the natural way. When~$I=\{q_0\}$, we use $I=\state_0$. An \nfa accepting a non-empty language is \emph{trim} if every state is accessible from an initial state and every state leads to a final state. Given a state $q\in Q$, the \emph{right language} of $\state$ is $\lang_\state(\aut{A}) = \{\,w \in \Sigma^\star \mid \delta(\state,w)\cap F \neq \emptyset\,\},$ and the \emph{left language} is $\overleftarrow{\lang}_\state(\aut{A}) = \{\,w \in \Sigma^\star \mid \state \in\delta(I,w)\,\}.$ The \emph{language accepted} by $\aut{A}$ is $\lang(\aut{A})=\bigcup_{\state\in I}\lang_{\state}(\aut{A})$. An \nfa $\aut{A}$ is \emph{minimal} if it has the smallest number of states among all \nfas that accept $\lang(\aut{A})$.

An \nfa is \emph{deterministic} (\dfa) if $|I|=1$ and $|\delta(\state,\letter)|\leq 1$, for all $(\state,\letter) \in Q\times\Sigma$. We can convert an \nfa $\aut{A}$ into an equivalent \dfa $\AD(\aut{A})$ using the well-known subset construction. Two states $\state_1$, $\state_2$ are \emph{equivalent} (or \emph{indistinguishable}) if~$\mathcal{L}_{\state_1}(\aut{A}) = \mathcal{L}_{\state_2}(\aut{A})$. A \emph{minimal} \dfa has no different equivalent states, every state is reachable and it is unique up to isomorphism.

The \emph{state complexity} of a language $L$, $\dsc(L)$, is the size of its minimal~\dfa. The \emph{nondeterministic state complexity} of a language $L$, $\nsc(L)$, is defined analogously. 

A trim \nfa $\aut{A}=\langle Q,\Sigma,\delta,I, F\rangle$ for a finite language of words of size at most $\ell$ is acyclic and ranked, i.e., the set of states  $Q$ can be partitioned into $\ell+1$ disjoint sets $Q_0\cup Q_1 \cup \cdots \cup Q_\ell$, such that for every state $q \in Q_i$, $\aut{A}$ reaches a final state by words of length at most~$i$ ($Q_i=\cset{\state\in Q\mid \forall w\in \Sigma^\star, \delta(\state,w)\in F \implies |w|\leq i}$) and all transitions from states of rank $i$ lead only to states in $i'$, with $i,i'\in[\ell]$ and $i'<i$. We define the \emph{width} of a rank~$i$, namely $\width(i)$, as the cardinality of the set $Q_i$, and the \emph{width} of~$\aut{A}$ to be the maximal width of a rank, i.e., $\width(\aut{A})=\max_{i\in[\ell]}|Q_i|$. A \dfa for a finite language is also ranked but it may have a \emph{dead state}~$\Omega$ which is the only state with a self-loop and without a rank. In a trim acyclic automaton, two states $\state$ and $\state'$ are equivalent if they are both in the same rank, either final or not final, and their transition functions lead to equivalent states, i.e., $\delta(\state,w) \in F \iff \delta(\state',w)\in F$, for each word~$w\in\Sigma^*$. An acyclic \dfa{} can be minimized by merging equivalent states and the resulting algorithm runs in linear time in the size of the automaton (Revuz algorithm, \cite{revuz92,almeida08}).

\subsection{Block Languages and Bitmap Representation} \label{sec:block}
Given an alphabet $\Sigma=\{\letter_0,\ldots,\letter_{k-1}\}$ of size $k>0$ and an integer $\ell > 0$,
a \emph{block language} $L \subseteq \Sigma^\ell$ is a set of words of length $\ell$ over $\Sigma$.
Let $\aut{A}=\langle Q ,\Sigma,\delta,\state_0,F\rangle$ be a $\nfa$ that accepts a block language with a single initial state.
Because all accepted words have the same length,
we can assume that the finite automata for block languages have only one final state
, i.e., $F=\{q_f\}$, for some~$q_f\in Q$.
Moreover, as before, the set of states $Q$ can be partitioned into~$Q_0\cup Q_1 \cup \cdots \cup Q_\ell$ where $Q_i$ is the set of states with rank~$i$ and $\delta(Q_{i},\letter)\subseteq Q_{i-1}$,
where $i=1,\ldots,\ell$ and $\sigma\in\Sigma$. We also have a unique final state in rank~0, that is~$Q_0=F=\{q_f\}$. If $A$ is a \dfa for a block language, then there exists an extra dead state $\Omega$. For each~$\state \in Q_i$ and $\letter\in \Sigma$, either $\delta(\state,\letter)\in Q_{i-1}$ or $\delta(\state,\letter)=\Omega$ (but $\state$ must have at least a transition to~$Q_{i-1}$), for all $i\in [1,\ell]$. 

Câmpeanu and Ho \cite{campeanu04} estimated the maximal number of states of a minimal \dfa accepting a block language.
In the next lemma, we recall that result and related properties.
In \cref{fig:ADFA} the constraints on the widths of the ranks of a minimal \dfa are depicted.

\begin{lemma}\label{lem:maxsc}
	Let $L\subseteq\Sigma^\ell$ be 	block language  over an alphabet of size $k$ and $\ell>0$. Then, we  have
	\begin{enumerate}
		\item \label{lem:maxsc:sc} $\dsc(L)\leq \frac{k^{\ell-r+1}-1}{k-1}+\sum_{i=0}^{r-1}(2^{k^i}-1)+1$, where $r=\min\{n\in[\ell]\mid k^{\ell-n}\leq  2^{k^{n}}-1\}$;
		\item \label{lem:maxsc:rk} $r=\lfloor \log_k\ell\rfloor + 1 + x$, for some $x\in\{-1,0,1\}$;
		\item \label{lem:maxsc:width}
			Let $A$ be a minimal \dfa of maximal size for a block language. Then, $\width(A)=\max\{k^{\ell-r},2^{k^{r-1}}-1\}$, where $\width(r-1)=2^{k^{r-1}}-1$ and $\width(r)=k^{\ell-r}$. Moreover, let $r_{k,\ell}$ be the rank that the width of $\aut{A}$ is reached, either $r$ or~$r-1$.
	\end{enumerate}
\end{lemma}
\begin{proof}[Proof (sketch).] 
	The first statement was proven in~\cite[Corollary 10]{campeanu04} and follows from the fact that for each rank $i\in [\ell]$, we have that $\width(i)\leq 2^{k^i}-1$ and $\width(\ell- i) \leq k^i$. Then, for a \dfa to have maximal size we have $\width(i)= 2^{k^i}-1$, for $i\in [r-1]$, and $\width(i)=k^{\ell-i}$, for $i\in [r,\ell]$. Finally, we need to add one for  the dead state. The second statement was proven in~\cite{dmpr24a}. The third statement follows from the first, noticing that $\width(r-1)=2^{k^{r-1}}-1$ and  depending on whether $k^{\ell-r} > 2^{k^{r-1}}-1$ or not. We set $r_{k,\ell}$ to be the rank such that $\width(A)=\width(r_{k,\ell})$.
\end{proof}

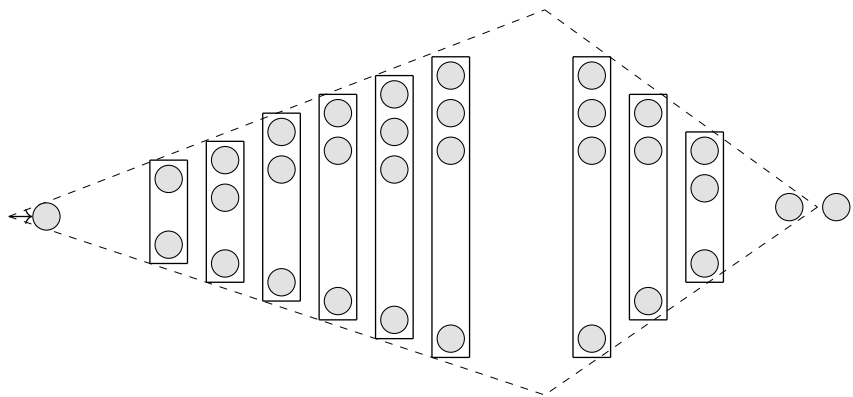
\begin{figure}[H]
	\centering\ctikzfig{ADFA}
	\vspace{-5cm}
	\caption{Constraints in the widths of the ranks of a minimal \dfa for a block language. Each rank (except the last and the first ones) is represented by a rectangle. The rightmost state is the dead-state ($\Omega$).
	}
	\label{fig:ADFA}
\end{figure} 

A block language $L$ can be characterized by a word in $\{0,1\}^{k^\ell}$ called \emph{bitmap} and denoted as 
$$\bs(L)= b_0\cdots  b_{k^\ell-1},$$ 
where $b_i=1$ if and only if $i\in [k^\ell-1]$ is the index of $w$ in the lexicographical ordered list of all the words of $\Sigma^\ell$ and the word $w$ is in $L$. In this case, we denote $i$ by~$\indi(w)$. The bitmap of a language can be denoted by $\bs$ when it is unambiguous to which language the bitmap refers to. Reciprocally, given a bitmap $\bs \in \{0,1\}^{k^\ell}$ and an alphabet $\Sigma$ of size $k$, $\lang(\bs)\subseteq \Sigma^\ell$ denotes the language represented by $\bs$.  

Boolean bitwise operations on bitmaps trivially correspond to Boolean set operations on block languages of the same length.
Formally,
given two bitmaps~$\bs_{1}, \bs_{2}\in \{0,1\}^{k^\ell}$,
the bitmap~$\bs_{1}\circ \bs_{2}$ is obtained by carrying out the bitwise operation~$\circ\in\{\lor,\land\}$ between $\bs_{1}$ and~$\bs_{2}$,
while~$\overline{\bs}_1$ is the bitwise complement of $\bs_{1}$.

Duarte et al.~\cite{dmpr24a} studied block languages using bitmaps. In particular, it was shown how to convert bitmaps into minimal deterministic and nondeterministic finite automata.

A bitmap $\bs \in \{0,1\}^{k^\ell}$ of a language $L\subseteq\Sigma^\ell$, for some $\ell>0$, can be split into factors of length $k^i$, for~$i\in [\ell]$. Let~$s^i_j=b_{jk^i}\cdots b_{(j+1) k^i-1}$ denote the $j$-th factor of length $k^i$, for $j \in [k^{\ell-i}-1]$. Since each factor of length $k^i$ can also be split into $k$ factors, $s^i_j$ is inductively defined as:
\begin{align*}
	s^i_j=
	\begin{cases}
		b_j, & \text{if } i = 0, \\
		s^{i-1}_{jk} \cdots s^{i-1}_{(j+1)k-1}, & \text{otherwise.}
	\end{cases}
\end{align*}

Furthermore, let $i\in [\ell]$, $j\in[k^{\ell-i}-1]$, and~$w\in\Sigma^{\ell-i}$ be the word of index $j$ of length~$\ell-i$, in lexicographic order. Then, $s^i_j$ corresponds to the bitmap of $w^{-1}L$.

Given a bitmap~$\bs\in\{0,1\}^{k^\ell}$, let $\bsset{i}$ be the set of factors of~$\bs$ of length~$k^i$, for~$i \in [\ell]$, in which there is at least one bit different from zero, that is, 
	$$\bsset{i} = \{s \in \{0,1\}^{k^i} \mid \exists j \in [k^{\ell-i}-1]: s = s^i_j\neq 0^{k^i} \}\text.$$

\begin{example}\label{ex:bitmap}
Let $\Sigma=\{a,b\}$, $k=2$, and $\ell=4$. Consider  $$L=\{aaaa,aaba,aabb,abab,abba,abbb,babb,bbaa,bbab,bbba\}.$$ The bitmap of $L$ is $\bs(L)=1011011100011110$. Moreover, we have that $s^2_0 = 1011$ is the bitmap of $(aa)^{-1}L=\{aa,ba,bb\}$, $s^3_1=00011110$  the bitmap of $b^{-1}L=\{abb, baa, bab, bba\}$, and $s^4_0=\bs$ the bitmap of~$L$. We also have $\bsset{0}=\{1\}$, $\bsset{1}=\{01,10,11\}$, $\bsset{2}=\{0001, 0111, 1011,$ $1110\}$, $\bsset{3}=\{00011110,10110111\}$, and $\bsset{4}=\{\bs\}$.
\end{example}

The sets $\bsset{i}$ are related to the states of the minimal finite automata representing the block language with bitmap $\bs$, as shown in~\cite{dmpr24a}. We now briefly recall such a result.

Given a bitmap $\bs$ associated with a block language $L \subseteq \Sigma^\ell$, with $|\Sigma|=k$ and~$\ell>0$, one can directly build the minimal \dfa $\aut{A}$ for $L$. Formally, $\aut{A}=\langle Q\cup \{\Omega\},\Sigma,\delta,\bs,\{1\}\rangle\text,$ where the set of states $Q$ correspond to bitmap factors, that is, $Q=\bigcup_{i\in[\ell]} \bsset{i}$; the initial state is the bitmap $\bs$; and the final state is the bitmap factor $1$. The transition function $\delta$ is given by the decomposition of each bitmap factor. Let~$s\in\bsset{i}$, where $s=s_0\cdots s_{k-1}$ and $|s_j|=k^{i-1}$, for $i\in[1,\ell]$ and $j\in[k-1]$. Then, the transition function contains $\delta(s,\letter_j)=s_j$. Moreover, the states in $\bsset{i}$ have rank $i$. The $\dfa$ can be completed considering transitions to $\Omega$ (dead-state) in the usual way.

A similar construction can be used to obtain a minimal \nfa for $L$, where each rank will contain the minimal cover of the sets $\{\bsset{i}\}_{i\in[\ell]}$. The main difference with the deterministic case is that the quotients of the language, corresponding to factors from the bitmap, are represented by a set of states, instead of a single one. For $i\in[\ell]$ and~$s\in\bsset{i}$, a \emph{cover} of $s$ is a set of~$n>0$ binary words~$\{c_0,\ldots, c_{n-1}\}$, where~$|c_j|=s$, for all~$j\in[n-1]$, such that the disjunction of the set equals $s$, that is, $\bigvee_{j\in [n-1]} c_j = s$. Since bitmap factors correspond to block languages, we have $\lang(s)=\bigcup_{j\in[n-1]}\lang(c_j)$. A set~$\cover_i$ of binary words of length~$k^i$ is a cover for the set~$\bsset{i}$ if all the words in~$\bsset{i}$ are covered by~$\cover_i$. For instance, it can be easily noticed that~$\bsset{i}$ covers itself. Moreover, we say that~$\cover_i$ is a \emph{minimal cover} for~$\bsset{i}$ if there is no other set~$\cover_i'$ smaller than~$\cover_i$ that covers~$\bsset{i}$. Then, a minimal \nfa $\aut{A}=\langle Q, \Sigma,\delta, \{\bs\}, \{1\}\rangle$ for~$L$ can be constructed as follows. As indicated, the single final state is the factor $1$. Additionally, we define the function~$\rho:\{0,1\}^\star\to 2^{\{0,1\}^\star}$ that maps factors into covers, where initially we set~$\rho(1)=\{1\}$. Next, for every rank~$i=1,\ldots,\ell$, we consider a minimal cover $\cover_i$ for~$\bsset{i}$, and we set, for every $s\in\bsset{i}$, $\rho(s)=\{c_0,\ldots,c_{n-1}\}\subseteq \cover_i$, such that~$\rho(s)$ covers~$s$. Furthermore, we set~$\cover_i$  as the set of states at rank~$i$ of~$\aut{A}$, and so~$Q=\bigcup_{i\in[\ell]}\cover_i$. The transitions from rank~$i$ to rank~$i-1$ will then be determined in a similar way to the~DFA construction. For each state~$c\in\cover_i$ in rank~$i$, we decompose~$c$ into~$c_0\cdots c_{k-1}$, where~$|c_j|=k^{i-1}$, for every~$j\in[k-1]$, and set~$\delta(c,\sigma_j)=\rho(c_j)$, if only~$c_j\neq0^{k^{i-1}}$, where~$\letter_j\in\Sigma$. We must also guarantee that~$\rho$ is defined for each~$c_j$ or, alternatively, that~$c_j\in\bsset{i-1}$. For that, we need to limit the search space of the cover~$\cover_i$, so that each word in the set is a composition of~$k$ words from~$\bsset{i-1}$ or~$0^{k^{i-1}}$. Formally,~$\cover_i\subseteq (\bsset{i-1}\cup 0^{k^{i-1}})^k\setminus 0^{k^i}$. Also,~$\bsset{\ell}=\{\bs\}$, so the minimal cover for~$\bsset{\ell}$ is itself. This result implies that~$\bs$ will be the single initial state at rank~$0$.


In  this paper, bitmaps will be a useful tool for the study of operational state complexities. Not only languages are easily represented by their bitmaps but also bitwise operations on bitmaps mimic the operations on languages.

\section{Operational Complexity} \label{sec:sc}
In this section we consider operations on block languages using their bitmap representations and study both the deterministic and nondeterministic state complexity of those operations. More precisely, the \emph{operational state complexity} is the worst-case state complexity of a language resulting from the operation, considered as a function of the state complexities of the operands. For instance, the state complexity of the union of two block languages can be stated as follows: given an $m$-state \dfa~$A_1$ and an $n$-state \dfa $A_2$, how many states are sufficient and necessary, in the worst case, to accept the language $L(A_1) \cup  L(A_2)$ by a \dfa? 

An upper bound can be obtained by providing an algorithm that, given \dfas for the operands, constructs a \dfa that accepts the resulting language, and the number of states, in the worst case, of the resulting \dfa is an upper bound for the state complexity of the referred operation. To show that an upper bound is tight, a family of languages (one language, for each possible value of the state complexity) for each operation must be given such that the resulting automata achieve that bound. We can call those families \emph{witnesses} or \emph{streams}.
 
We will mainly consider operations under which the family of block languages is closed, i.e., the resulting language is also a block language. In particular, we will consider the union and intersection of two block languages whose words are of the same length, the concatenation of two arbitrary block languages, the reversal, the complement of block languages closed to the block (i.e.,~$\Sigma^\ell \setminus L$), and word addition and removal from a block language. We will also analyze the Kleene star and plus operations of block languages, which in general do not yield a block language.


Of course, the upper bounds of operational state complexity known for finite languages apply for block languages. In \cref{tab:cfin}, we review some complexity results for finite languages. The first two lines give the bounds for the determinization of an $m$-state \nfa and the asymptotic upper bound of the maximal size of a minimal \dfa, respectively. For the operational state complexities, we consider $|\Sigma|=k$ or $|\Sigma|=f(\overline{m})$ if a growing alphabet is taken into account, $|F_i|=f_i$, and $p_i=|F_i-\{\state_i\}|$, for the $i$-th operand and its set of final states,~$F_i$.

Additionally, we show how to build the bitmap of the language resulting by applying each operation and also present a family of witness languages parameterized by the state complexity of the operands to show that the bounds provided are tight. In general, other additional parameters are the length $\ell$ of the words and the widths of each rank.

\begin{table} 
	\caption{Some complexity bounds for finite languages} 
	\centering
	\begin{tabular}{lcccccc}
		\toprule
		& \multicolumn{1}{c}{Upper bound} & \multicolumn{1}{c}{$|\Sigma|$} & Ref. \\
		\midrule
		\nfa $\to$ \dfa & $\Theta(k^{\frac{m}{1+\log k}})$ & $2$ & \cite{salomaa97} \\
		$\dsc(L)$ & $\frac{k^{\ell+2}}{\ell(k-1)^2\log_2{k}}(1+o(1))$ & $2$ & \cite{campeanu04} \\
		\midrule
		& \multicolumn{1}{c}{$\dsc$} & \multicolumn{1}{c}{$|\Sigma|$} & Ref. & \multicolumn{1}{c}{$\nsc$} & \multicolumn{1}{c}{$|\Sigma|$} & Ref.\\
		\midrule
		$L_1\cup L_2$ & $mn-(m+n)$ & $f(m,n)$ & \cite{han08} & $m+n-2$ & $2$ & \cite{holzer03} \\
		$L_1\cap L_2$ & $mn - 3(m + n) + 12$ & $f(m,n)$ & \cite{han08} & $O(mn)$ & $2$ & \cite{holzer03} \\
		$\overline{L}$ & $m$ & $1$ & & $\Theta(k^{\frac{m}{1+\log k}})$ & $2$ & \cite{holzer03} \\ \addlinespace[2mm]
		\multirow{2}{*}{$L_1L_2$} &$(m-n+3)2^{n-2}-1$, $m+1\geq n$ & $2$ & \cite{campeanu01} & \multirow{2}{*}{$m+n-1$} & \multirow{2}{*}{$2$} & \cite{holzer03} \\
		& $m+n-2$, if $p_1=1$ & $1$ & & \\ \addlinespace[2mm]
		\multirow{2}{*}{$L^\star$} & $2^{m-3}+2^{m-p-2}$, $p\geq 2$, $m\geq 4$ & $3$ & \cite{campeanu01} & \multirow{2}{*}{$m-1$, $m>1$} & \multirow{2}{*}{$1$} & \cite{holzer03} \\ & $m-1$, if $f=1$ & $1$ & & \\ \addlinespace[2mm]
		$L^{R}$ & $O(k^{\frac{m}{1+\log k}})$ & $2$ & \cite{campeanu01} & $m$ & $2$ & \cite{holzer03} \\
		\bottomrule
	\end{tabular}\label{tab:cfin}
\end{table}

\subsection{Reversal}
In the following, given a bitmap $\bs$ of a block language $L\subseteq \Sigma^\ell$, $|\Sigma|=k$, and $\ell>0$, we compute the bitmap for the reversal language $\R{L}$. Recall that the perfect shuffle of length $1$, denoted~$\shuffle_1$, of two words $u=u_0\cdots u_{n-1}$ and $v=v_0\cdots v_{n-1}$ of the same length is obtained by interleaving the letters of~$u$ and~$v$, namely $u\shuffle_1 v= u_0v_0\cdots u_{n-1}v_{n-1}$. If $j$ is a divisor of $n$, the perfect shuffle of length $j$, denoted~$\shuffle_j$, of $u$ and $v$ is the perfect shuffle of blocks of length $j$, that is, 
	$$u \shuffle_j v = u_0\cdots u_{j-1}v_0\cdots v_{j-1} \cdots u_{n-j}\cdots u_{n-1} v_{n-j}\cdots v_{n-1}\text.$$
	Finally, if $|u|=|v|$, we denote $u \shuffle_j v$ by $\shuffle^2_j(uv)$. This can be generalized for any number~$m\geq 2$ of words~$w_0,\ldots, w_{m-1}$ of the same length by considering the perfect shuffle of blocks of length~$j$ taken from each of the $w_i$ words, that is, $\shuffle_j^m(w_0\cdots w_{m-1})$. For $j=1$ and $w_i=w_{i,0}\cdots w_{i,n-1}$, $i\in [m-1]$, one has 
		$$ \shuffle_1^m(w_0\cdots w_{m-1})=w_{0,0}\cdots w_{n-1,0}w_{0,1}\cdots w_{m-1,1}\cdots w_{0,n-1}\cdots w_{m-1,n-1}.$$
Let $\mathsf{R}_0=\bs$ and $\mathsf{R}_{i}=\shuffle^{k}_{k^{i-1}}(\mathsf{R}_{i-1})$, for $i\in [1, \ell -1]$ and $k=|\Sigma|$.
\begin{lemma}
	Let $L\subseteq\Sigma^\ell$ be a block language, for some $\ell>0$. The bitmap for the reversal of $L$, namely $\bs(\R{L})$, is~$\mathsf{R}_{\ell-1}$.
\end{lemma}
\begin{proof}
	Let us prove that~$\lang(\mathsf{R}_{\ell-1})=\R{L}$. For~$i=0$, of course~$\lang(\mathsf{R}_0)=\lang(\bs)=L$. Next, for~$i=1$ we have~$\mathsf{R}_1=\shuffle^k_1(\bs)$. This operation performs the cyclic permutation~$S_1=(0\;1\;\cdots\;\ell-1)$ in each word of~$\lang(\bs)$, that is, each symbol of every word in~$\lang(\bs)$ is shifted one position to their right and the last symbol becomes the first. The following operation,~$\mathsf{R}_2=\shuffle^k_k(\mathsf{R}_1)$, performs the permutation~$S_2=(1\;2\;\cdots\;\ell-1)$ in every word of~$\lang(\mathsf{R}_1)$. Analogously, in this transformation each symbol apart from the first of every word in~$\lang(\mathsf{R}_1)$ is shifted one position to their right but the last symbol becomes the second. In general, the~$j$-th shuffle performs the permutation~$S_j=(j-1\;j\;\cdots\;\ell-1)$, for $j\in [1,\ell-1]$. The composition of the transformations~$S_1,S_2,\ldots,S_{\ell-1}$ ensure that~$\lang(\mathsf{R}_{\ell-1})=\R{L}$ \cite{Diaconis:1983aa}.
\end{proof}

\begin{example}
	Let~$\Sigma=\{a,b\}$ and~$\ell=3$. Let~$\bs=b_0b_1b_2b_3b_4b_5b_6b_7$ be a bitmap for a block language~$L$ such that~$b_0=b_3=b_4=1$, and the remaining bits are~$0$. We have 
	\begin{align*}
		\mathsf{R}_0 & = b_0b_1b_2b_3b_4b_5b_6b_7 \quad \text{ and } \quad \lang(\mathsf{R}_0) = \{aaa, abb, baa\}, \\
		\mathsf{R}_1 & = b_0b_4b_1b_5b_2b_6b_3b_7 \quad \text{ and } \quad \lang(\mathsf{R}_1) = \{aaa, bab, aba\}, \\
		\mathsf{R}_2 & = b_0b_4b_2b_6b_1b_5b_3b_7 \quad \text{ and } \quad \lang(\mathsf{R}_2) = \{aaa, bba, aab\}, \\
	\end{align*}
	and~$\lang(\mathsf{R}_2)=\R{L}$, as desired.
\end{example}

Now we turn to the analyze of the state complexity of this operation. The \dfa for the reversal of a block language $L\subseteq\Sigma^\ell$, with $\ell>0$, is given by reversing each transition on a \dfa for $L$ and then determinising the resulting \nfa. The cost of the determinisation of an $m$-state \nfa for a block language is~$2^{\Theta(\sqrt{m})}$ in size~\cite{KarhOkho:2O14}, so the state complexity of the reversal must also be limited by this bound. 

\begin{corollary}\label{cor:upperscrev}
	Given an $m$-state \dfa for a block language $L$, $2^{O(\sqrt{m})}$ states are sufficient for a \dfa accepting $\R{L}$.
\end{corollary}

In the following, we show that this bound is tight. Let $\ell>0$, $\Sigma=\{a,b\}$, $k=2$, and consider \cref{lem:maxsc}. We define a family of languages, parametrized by $\ell$, that attain the maximal state complexity. For convenience, let $r_\ell=r_{2,\ell}$. Then, consider
	$$\maxmindfa_{\ell} = \{ w_1w_2 \mid w_1\in\Sigma^{\ell-{r_\ell}},\, w_2\in\Sigma^{r_\ell},\, i=\indi(w_1),\, j = \indi(w_2),~\text{and~}  (i+1) \wedge 2^j\neq 0\}\text,$$
where we use the notation~$i\land2^j\neq0$ to indicate that the $j$-th least significant bit of the binary representation of $i$, namely $i_{[2]}$, is $1$. Informally, these languages contain words of size~$\ell$ that can be split in~$w_1$ of size~$\ell-\maxr$ and~$w_2$ of size~$\maxr$, with corresponding indices~$i=\indi(w_1)$ and~$j=\indi(w_2)$, such that the $j$-th least significant bit of~$\bin{(i+1)}$ is~$1$.

\begin{proposition}[\cite{dmpr24a}]\label{lemma:bitmap-of-max}
	The minimal \dfa $A$ accepting the language $\maxmindfa_{\ell}$ has maximal size and $\m=\width(A)=\width(r_{\ell})=\max(2^{\ell-r}, 2^{2^{r-1}}-1)$. Moreover, let
		$$ P_{\m,\maxr} = \prod_{i=1}^{\m} \R{\pad(i_{[2]}, 2^{r_\ell})}\text,$$
	where $\pad(s,t)$ is a function that adds leading zeros to a binary string $s$ until its length equals $t$. Then the bitmap of the language~$\maxmindfa_{\ell}$ is given by
	\begin{align*}
		\bs(\maxmindfa_{\ell}) =
		\begin{cases}
			P_{\m,\maxr},                      & \text{if } \m=2^{\ell-r}, \\
			P_{\m,\maxr} \cdot 0^{2^{r_\ell}}, & \text{if } \m=2^{2^{r-1}}-1.	
		\end{cases}
	\end{align*}
\end{proposition}

\begin{example}\label{ex:max}
	For $\ell=5$, we have 
	\begin{align*}
		\maxmindfa_5 = &\{aaaaa, aabab, abaaa, abaab, abbba, baaaa, \\
		       	       & baaba, babab, babba, bbaaa, bbaab, bbaba, bbbbb\},
	\end{align*}
	and $\bs(\maxmindfa_{5})= \prod_{i=1}^{8} \R{\pad(i_{[2]}, 4)}=10000100110000101010011011100001$. The minimal \dfa is represented below, where the dead state is omitted.
	
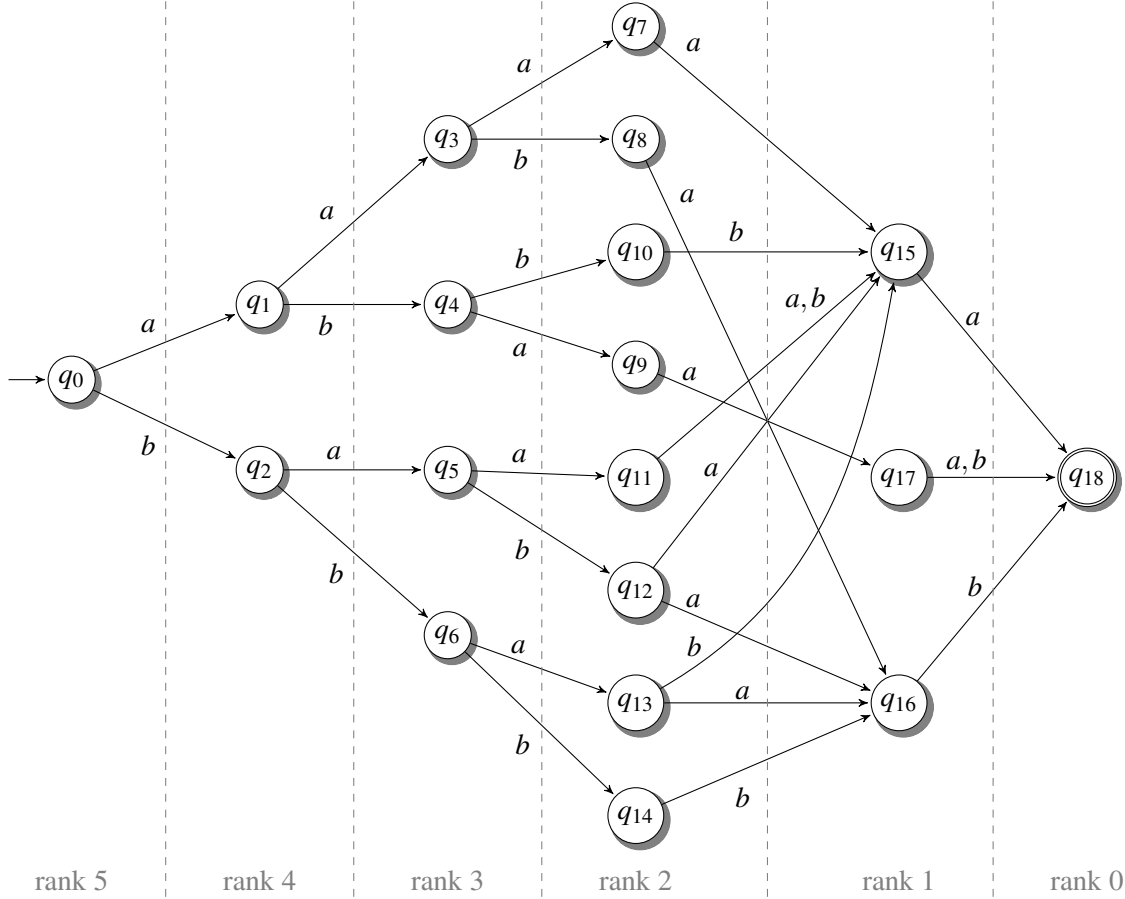
\begin{figure}[ht]
	\centering
	\begin{tikzpicture}[>=stealth', shorten >=1pt, auto, node distance=1.5cm, initial text={}]
     	\node[sstate]           (7)                                      {$q_{7}$};
     	\node[sstate]           (8)  [below of=7]                        {$q_{8}$};
     	\node[sstate]           (10) [below of=8]  				         {$q_{10}$};
     	\node[sstate]           (9)  [below of=10]                       {$q_{9}$};
     	\node[sstate]           (11) [below of=9]				         {$q_{11}$};
     	\node[sstate]           (12) [below of=11]  			         {$q_{12}$};
     	\node[sstate]           (13) [below of=12]          	         {$q_{13}$};
     	\node[sstate]           (14) [below of=13]				         {$q_{14}$};
     	\node[sstate]           (3)  [left of=8,xshift=-1cm]	         {$q_{3}$};
     	\node[sstate]           (4)  [below of=3,yshift=-.7cm]	         {$q_{4}$};
     	\node[sstate]           (5)  [below of=4,yshift=-.7cm]	         {$q_{5}$};
     	\node[sstate]           (6)  [below of=5,yshift=-.7cm]	         {$q_{6}$};
     	\node[sstate]           (1)  [left of=4,xshift=-1cm]	         {$q_{1}$};
     	\node[sstate]           (2)  [left of=5,xshift=-1cm]	         {$q_{2}$};
     	\node[sstate,initial]   (0)  [left of=1,xshift=-1cm,yshift=-1cm] {$q_{0}$};
     	\node[sstate]           (15) [right of=10,xshift=2cm]	         {$q_{15}$};
     	\node[sstate]           (17) [right of=11,xshift=2cm]	         {$q_{17}$};
      	\node[sstate]           (16) [right of=13,xshift=2cm]	         {$q_{16}$};
     	\node[sstate,accepting] (18) [right of=17,xshift=1cm]	         {$q_{18}$};
      	\path[->] 
        	(0)  edge                node                       {$a$}   (1)
          	     edge                node[swap]                 {$b$}   (2)
        	(1)  edge                node[pos=.44]              {$a$}   (3)
          	     edge                node[swap,pos=.3]          {$b$}   (4)
        	(2)  edge                node[pos=.36]              {$a$}   (5)
          	     edge                node[swap]                 {$b$}   (6)
        	(3)  edge                node                       {$a$}   (7)
          	     edge                node[swap,pos=.35]         {$b$}   (8)
        	(4)  edge                node[swap,pos=.47]         {$a$}   (9)
          	     edge                node                       {$b$}   (10)
        	(5)  edge                node[pos=.35]              {$a$}   (11)
          	     edge                node[swap]                 {$b$}   (12)
        	(6)  edge                node[inner sep=1,pos=.27]  {$a$}   (13)
          	     edge                node[swap]                 {$b$}   (14)
        	(7)  edge                node[pos=.1]               {$a$}   (15)
       	 	(8)  edge                node[pos=.1]               {$a$}   (16)
        	(9)  edge                node[pos=.1,inner sep=1]   {$a$}   (17)
        	(10) edge                node[pos=.35]              {$b$}   (15)
        	(11) edge                node[pos=.77,inner sep=.2]   {$a,b$} (15)
        	(12) edge                node[pos=.3,inner sep=1]   {$a$}   (15)	
                 edge                node[pos=.1,inner sep=1]   {$a$}   (16)
        	(13) edge[out=30,in=260] node[pos=.1,inner sep=1]   {$b$}   (15)
          		 edge                node[pos=.38,inner sep=1]  {$a$}   (16)
        	(14) edge                node[swap] [pos=.3]        {$b$}   (16)
        	(15) edge                node[inner sep=1,pos=.3]   {$a$}   (18)
        	(16) edge                node[inner sep=1,pos=0.45] {$b$}   (18)
        	(17) edge                node[inner sep=1,pos=.3]   {$a,b$} (18)  
        	;
		\path[gray]
			(14)++(-90:5ex) node (rank 2) {rank~$2$}
			(rank 2 -| 6)   node (rank 3) {rank~$3$}
			(rank 2 -| 2)   node (rank 4) {rank~$4$}
			(rank 2 -| 0)   node (rank 5) {rank~$5$}
			(rank 2 -| 16)  node (rank 1) {rank~$1$}
			(rank 2 -| 18)  node (rank 0) {rank~$0$}
		;
		\foreach \r [remember=\r as \lastr (initially 0)] in {1,...,5} {
			\draw[dashed,gray] ($(rank \r)!.5!(rank \lastr)$)+(0,11.7cm) -- ($(rank \r.south)!.5!(rank \lastr.south)$);
		}
	\end{tikzpicture}
	\caption{The minimal DFA accepting the language~$\maxmindfa_{\ell}$ for~$\ell=5$. The sink-state is omitted, as well as all transitions from and to it.}
	\label{figure:dfa-witness}
\end{figure}

\end{example}

The reversal of $\maxmindfa_{\ell}$, namely $\R{\maxmindfa_{\ell}}$, has a minimal \dfa whose width is at most $2^{r_\ell+1}$.
\begin{lemma} \label{lemma:reversal-maxrank}
	Let $r$ and $r_\ell$ be defined as before for $\ell>0$ and alphabet size $k=2$. A minimal \dfa $\aut{B}$ such that $\width(\aut{B})\leq 2^{r_\ell+1}$ is sufficient to accept the reversal of the language~$\maxmindfa_{\ell}$.
\end{lemma}
\begin{proof}	
	Let $\aut{B}$ be a \dfa such that the last $r_\ell+1$ ranks have maximal width, that is, $2^{\ell-r'}$ states in each rank $r' \in [\ell-r_{\ell}, \ell]$. In particular, the width of the rank $\ell-r_\ell$ is $2^{r_\ell}$. We can order the states in this rank in such way that $\state_j$ is the state whose left language is the reverse of the $j$-th word of $\Sigma^{r_\ell}$, i.e., $\overleftarrow{\lang}(\state_j)=\{\R{w}\}$ where $j=\indi(w)$, for each $j\in[2^{r_\ell}-1]$. 
	 
	 Moreover, let us define the right language of $q_j$ as
		$$\lang_{q_j} = \cset{w\in\Sigma^{\ell-r_\ell}\mid i=\indi(\R{w}), \text{ }(i+1)\wedge2^j\neq 0}\text.$$ 
		Clearly, $\aut{B}$ accepts $\R{\maxmindfa_{\ell}}$. Consider $\aut{B_{q_j}}$, the \dfa $\aut{B}$ with initial state $q_j$, and let us show that $\width(\aut{B_{q_j}}) \leq 2$.
		
	Let $j\in[2^{r_\ell}-1]$, $r' \in [\ell-r_\ell-1]$, and $\bsr_{r'}$ be the set of factors of size $2^{r'}$ of the bitmap of $\R{\maxmindfa_{\ell}}$. Let $s\in\bsr_{r'}$ and $w_1\in\Sigma^{\ell-r_\ell-r'}$ such that $s$ represents the language $w_1^{-1}\lang_{q_j}$. From the definition of $\lang_{q_j}$, $w_1w_2\in \lang_{q_j}$ if $(i+1)\wedge2^j\neq 0$, where $w_2\in\Sigma^{r'}$ and $i=\indi(\R{w_2}\R{w_1})$. We will now show that the number of states on the rank $r'$ of $\aut{B_{q_j}}$ is bounded by $2$, by arguing that $|\bsr_{r'}| \leq 2$. Consider the following two cases:
		\begin{enumerate}
			\item $j < \ell-r_\ell-r'$: in this case, it is sufficient to check if $(i'+1)\wedge2^j\neq0$, where $i'=\indi(\R{w_1})$, since $|\R{w_1}|=\ell-r_\ell-r'$. Then, either $i'+1$ has its $j$-th bit equal to $1$, which implies that $s=1\cdots 1$, or has not, implying that $s=0\cdots 0$, so $s\notin \bsr_{r'}$. Therefore, $|\bsr_{r'}| = 1$.

			\item $j \ge \ell-r_\ell-r'$: if $w_1\in\Sigma^{\ell-r_\ell-r'}\setminus \{b^{\ell-r_\ell-r'}\}$, the binary representation of $i'+1$, with $i'=\indi(\R{w_1})$, requires at most $\ell-r_\ell-r'$ bits. Then, the $j$-th bit of $i$, corresponds to the ($j-\ell+r_\ell+r'$)-bit of~$w_2$. Hence, $u^{-1}\mathcal{L}_{q_j} = v^{-1}\mathcal{L}_{q_j}$, for all $u, v \in \Sigma^{\ell-r_\ell-r'}\setminus \{b^{\ell-r_\ell-r'}\}$. On the other hand, if $w_1=b^{\ell-r_\ell-r'}$, it results on a different quotient, since $\ell-r_\ell-r'+1$ bits are needed for the binary representation of $i'+1$. Therefore, $|\bsr_{r'}| = 2$.
		\end{enumerate}
	This result implies that~$\width(\aut{B}_{q_j})=2$. As a consequence, the width of the ranks~$r'\in [\ell-r_\ell-1]$ of~$\aut{A}$ are bounded by~$2^{\maxr+1}$, as desired.	
\end{proof}

Then, we have the following bound on the state complexity of the reversal of a language.
\begin{theorem}
	Let $L\subseteq\Sigma^\ell$, $\ell>0$, such that $\dsc(L)=m$. Then, $\dsc(\R{L})\in 2^{\Theta(\sqrt{m})}$. 
\end{theorem}
\begin{proof}
	By \cref{cor:upperscrev} we have that $2^{O(\sqrt{m})}$ states are sufficient for a \dfa accepting $\R{L}$. Now, we prove that this cost is necessary in the worst case.
	
	Let $\aut{A}$ and $\aut{B}$ be the minimal \dfas for $\maxmindfa_{\ell}$ and $\R{\maxmindfa_{\ell}}$ with set of states $Q$ and $P$, respectively. For this proof, assume that~$2^{\ell-r} > 2^{2^{r-1}}-1$, which implies that $r_\ell=r$. A similar proof follows, otherwise.
	
	 By \cref{lemma:bitmap-of-max}, $\aut{A}$ is of maximal size and its width is exactly~$2^{\ell-r}$.
	 Therefore, the number of states of $\aut{A}$ is
	\begin{align*}
		|Q| & = \sum_{i=0}^{r_\ell-1}(2^{2^i}-1) + \sum_{i=r_\ell}^{\ell}2^{\ell-i} = \sum_{i=0}^{r_\ell-1}2^{2^i} - r_\ell + 2^{\ell-r_\ell+1} - 1 \\
		    & \geq 2^{2^{r_\ell}-1} - (r_\ell+1) \tag*{$2^{\ell-r_\ell+1} > 0$} \\
		    & \geq 2^{2^{\log_2\ell+2}-1} - (\log_2\ell+3) \tag*{$2^{2^{r_\ell}-1} \gg r_\ell$, $r=r_\ell$, and $r<\log_2\ell+2$} \\
			& \geq 2^{4\ell-1} - (\log_2\ell+3) \in 2^{\Omega(\ell)}. \\
	\end{align*}
	By \cref{lemma:reversal-maxrank}, the width of $\aut{B}$ is bounded by $2^{r_\ell+1}$. Moreover, the width of each rank~$i\in[r-1]$ of~$\aut{B}$ is also bounded by~$2^{2^i}-1$, as we have seen in~\cref{lem:maxsc}. Then, let~$r'=\min\{n\in[\ell] \mid 2^{r_\ell+1} \leq 2^{2^n}-1\}$. In particular, we have
	\begin{align*}
		2^{r_\ell+1} > 2^{2^{r'-1}}-1 & \implies 2^{r_\ell+1} \geq 2^{2^{r'-1}} \implies r_\ell+1 \geq 2^{r'-1} \\
									  & \implies \log_2\ell + 3 \geq 2^{r'-1} \tag*{$r=r_\ell$ and $r<\log_2\ell+2$} \\
									  & \implies r' \leq \log_2(\log_2\ell +3) + 1.
	\end{align*}
	
	The value of~$r'$ tells us how many ranks in~$\aut{B}$ can achieve the maximal width of~$2^{r_\ell+1}$. Then, the number of states of $\aut{B}$ is bounded by
	\begin{align*}
		|P| & \leq \sum_{i=0}^{r'-1}(2^{2^i}-1) + 2^{r_\ell+1}(\ell-r_\ell-r') + \sum_{i=\ell-r_\ell}^{\ell}2^{\ell-i} \\
			& \leq 2^{2^{r'}} - r' + 2^{r_\ell+1}(\ell-r_\ell-r'+1) - 1 \\ 
			& \leq 2^{2^{\log_2(\log_2\ell +3) + 1}} + 2\cdot2^{r_\ell}(\ell+1) \tag*{$r' \leq \log_2(\log_2\ell +3) + 1$} \\
		    & \leq 2^6\cdot2^{\log_2\ell^2} + 2^3\cdot2^{\log_2\ell}(\ell+1) \tag*{$r=r_\ell$ and $r<\log_2\ell+2$} \\
			& \leq 2^6\ell^2 + 2^3(\ell^2+\ell) = O(\ell^2). \\
	\end{align*}
	
	Thus, given $L=\R{\maxmindfa_{\ell}}$ with $\dsc(L)=m$, we have that $\dsc(\R{L})=2^{\Omega\left(\sqrt{m}\right)}$, as desired.
\end{proof}

The \nfa for the reversal of a language $L\subseteq\Sigma^\ell$ is given by reversing the transitions on the \nfa for $L$. In fact, the nondeterministic state complexity of the reversal of a finite language coincides with the nondeterministic state complexity of the language, so no better result can be obtained for the block languages.
\begin{theorem}
	Let $L\subseteq\Sigma^\ell$, for some $\ell>0$. Then, $\nsc(\R{L}) =\nsc(L)$.
\end{theorem}
\begin{proof}
	The construction above shows that $\nsc(\R{L}) \leq \nsc(L)$. The following family of languages shows that it is tight. Let $L_\ell = \{a^\ell\}$, with $\ell>0$. We have both that $\nsc(L_\ell) = \ell+1$ and $L_\ell=\R{L_\ell}$.
\end{proof}

\subsection{Word Addition and Word Removal}
Consider a language $L\subseteq\Sigma^\ell$, for some $\ell>0$, over an alphabet of size $k$. The operations of adding or removing a word $w \in \Sigma^\ell$ from the language,  $L\setminus\{w\}$ and $L\cup\{w\}$, respectively, correspond to the not operation on the~$\indi(w)$-th bit of $\bs(L)$. From that observation, we can estimate the state complexity of these operations.

\begin{theorem} \label{theorem:sc-word-op}
	Let $L\subseteq \Sigma^\ell$ be a block language with $|\Sigma|=k$ and~$\ell>0$, such that $\dsc(L) = m$. Let~$\oplus \in \{\setminus, \cup\}$, $w \in \Sigma^\ell$, and~$L'=L \oplus \{w\}$. Then, $n - (\ell-1) \leq \dsc(L') \leq m + (\ell-1)$. 
\end{theorem}
\begin{proof}
	Let $\bs(L)$ and $\bs(L')$ be the bitmaps of $L$ and $L'$, respectively. Let us assume that the operation~$\oplus$ results in a different language. Then, the bitmaps~$\bs(L')$ and~$\bs(L)$ differ exactly for one bit. Let $i\in[\ell]$. Then, there is exactly one $j \in [k^{\ell-i}-1]$ such that $s^i_j \neq {t^i_j}$, where $s^i_j$ and ${t^i_j}$ denote the $j$-th bitmap factor of size $k^i$ of $\bs(L)$ and $\bs(L')$, respectively. Also, recall $\bsr(L)_i$ (resp. $\bsr(L')_i$), the set of factors of size $k^i$ of $\bs(L)$ (resp. $\bs(L')$). Then, there are four possible cases:
	\begin{enumerate}
		\item $s^i_j \in \bsr(L')_i$ and $t^i_j \in \bsr(L)_i$: the two sets have the same size;
		\item $s^i_j \in \bsr(L')_i$ and $t^i_j \notin \bsr(L)_i$: $\bsr(L')_i$ has one more element than $\bsr(L)_i$; \label{item:sc-word-op-increase}
		\item $s^i_j \notin \bsr(L')_i$ and $t^i_j \in \bsr(L)_i$: $\bsr(L)_i$ has one more element than $\bsr(L')_i$; \label{item:sc-word-op-decrease}
		\item $s^i_j \notin \bsr(L')_i$ and $t^i_j \notin \bsr(L)_i$: the two sets have the same size.
	\end{enumerate}
Then, the difference on the number of states from a \dfa which accepts the language $L'$ and the \dfa which accepts $L$ is bounded by $\ell-1$, which is the number of ranks neither initial nor final.
\end{proof}

These bounds also extend to the nondeterministic state complexity, as proved in the following result.
\begin{theorem} \label{theorem:nsc-word-op}
	Let $L\subseteq \Sigma^\ell$ be a block language with $|\Sigma|=k$ and~$\ell>0$, such that $\nsc(L) = m$. Let~$\oplus \in \{\setminus, \cup\}$, $w \in \Sigma^\ell$, and~$L'=L \oplus \{w\}$. Then, $m - (\ell-1) \leq\nsc(L') \leq m + (\ell-1)$.
\end{theorem}
\begin{proof}
	Consider the proof of Theorem~\ref{theorem:sc-word-op} and its notation. If the case~\ref{item:sc-word-op-increase} verifies, that is, $|\bsr(L')_i| = |\bsr(L)_i|+1$, then the cover will require at most one more segment to cover the new set. Analogously, the size of the cover for $\bsr(L')_i$ can be smaller by one than the cover for $\bsr(L)_i$, for the case~\ref{item:sc-word-op-decrease}.
\end{proof}

The upper bounds from Theorems~\ref{theorem:sc-word-op} and~\ref{theorem:nsc-word-op} are reached as stated in the following theorem.

\begin{theorem}\label{theorem:witness-word-op}
	The bounds given in \cref{theorem:sc-word-op,theorem:nsc-word-op} are tight.
\end{theorem}
\begin{proof}
	Let $\Sigma = \{a,b\}$ and $\ell>0$. Consider $L_\ell=\{a,b\}^\ell$, whose bitmap is $\bs(L_\ell) = 1^{2^\ell}$. Let $w=a^\ell$. We have that~$\nsc(L_\ell)=\ell+1$ and $\dsc(L_\ell)=1+\nsc(L)=2+\ell$, while~$\dsc(L_\ell\setminus\{w\})=2\ell+1=1+\nsc(L_\ell \setminus\{w\}) $. In the same way it is possible to prove for word addition.
\end{proof}

The family of languages in the previous proof is also a witness for the upper bound of the operation~$\Sigma^\ell\setminus\{w\}$.

\subsection{Intersection} \label{section:intersection}
Let $L_1, L_2 \subseteq \Sigma^\ell$ be two block languages, for some $\ell>0$ and $|\Sigma| = k$, and their respective bitmaps $\bs(L_1)$,~$\bs(L_2)$. The bitmap of the intersection of $L_1$ and $L_2$  is given by $\bs(L_1)\wedge\bs(L_2)$.

Now, let~$\aut{A}=\langle Q\cup \{\Omega_1\},\Sigma,\delta_1,\state_0,\{q_f\}\rangle$ and~$\aut{B}=\langle P\cup \{\Omega_2\},\Sigma,\delta_2,p_0,\{p_f\}\rangle$ be the minimal \dfas for $L_1$ and  $L_2$, respectively. For obtaining a \dfa for $L_1\cap L_2$, one can use the standard product construction, and obtain a product automaton  $\aut{C}$. As shown in~\cite{han08} (see Table~\ref{tab:cfin}) the size of $\aut{C}$ is at most $mn-3(m+n)+12$, if $|Q|=m-1$ and $|P|=n-1$. This bound is the result of:
\begin{itemize}
	\item There are no transitions to the initial state neither in $\aut{A}$, $\aut{B}$, nor $\aut{C}$ (this saves $m-1+n-1$ states);
	\item In $\aut{C}$, all pairs of states $(q,\Omega_2)$ and $(\Omega_1,p)$, for $q\in Q$ and $p\in P$, can be merged with $(\Omega_1,\Omega_2)$  (this saves $m-2+n-2$ states);
	\item In $\aut{C}$, all pairs of states $(q,p_f)$ or $(q_f,p)$, for $q\in Q$ and $p\in P$, can be merged with $(q_f,p_f)$ or  $(\Omega_1,\Omega_2)$ (if in general $q_f$ and $p_f$ are the pre-dead states,  this saves $m-3+n-3$ states). 
\end{itemize}

However, for block languages ,pre state can be saved since
a state~$(q, p)$ of $\aut{C}$ is both accessible from the initial state and leads to the final state if and only if~$\rank(q) = \rank(p)$, for every $q\in Q$, $p\in P$.  

Let $Q_i$ be the set of states in rank $i$ in $\aut{A}$ and $m_i=\width(i)=|Q_i|$, for $i\in [\ell]$. Let $P_i$ be the set of states in rank $i$ in $\aut{B}$ and $n_i=\width(i)=|P_i|$, for  $i\in [\ell]$. Additionally, $m=1+\sum_{i\in[\ell]}m_i$ and $n=1+\sum_{i\in[\ell]}n_i$  since the dead states $\Omega_j$ do not belong to any rank. We have that:

\begin{lemma} \label{lemma:s-intersection-sufficient}
	Given two \dfas $\aut{A}$ and $\aut{B}$ for block languages $L_1$ and $L_2$, respectively, a \dfa with $\sum_{i=0}^{\ell}m_in_i+1$ states is sufficient to recognize the intersection of $L_1$ and $L_2$, where $m_i$ and $n_i$ are the widths of rank $i$ in $\aut{A}$ and $\aut{B}$, respectively, for $i\in [\ell]$.
\end{lemma}
\begin{proof}
	Given the above considerations, the states of the \dfa resulting from trimming $\aut{C}$ are, in the worst-case, $\bigcup_{i\in[\ell]}Q_i\times P_i$ and a single dead state is needed.
\end{proof}

Let us show that this bound is tight for a fixed size of the alphabet, as opposed to the general case of finite languages where a growing alphabet is required~\cite{han08}. Consider the following family of languages, defined over an alphabet $\Sigma$ of size $k$, and let $d>0$ and $x\in\{0,1\}$:
	$$L_{k, d, x} = \cset{a_0\cdots a_{2d-1} \in \Sigma^{2d} \mid \forall i\in[d-1] : i \equiv x \Mod 2 \implies a_i = a_{2d-i}}.$$
Informally, it contains the words that can be split into two halves of size $d$, where, if~$x=0$ ($x=1$, resp.), then the symbols in even (odd, resp.) positions of the first half are equal to their 
	symmetric position
	in the second half.

\begin{lemma} \label{lemma:intersection-dfa-storage}
	Let~$k\geq2$,~$d\geq 0$, and~$x\in\{0,1\}$. Also, let~$\aut{A}$ be the minimal~DFA for~$L_{k, d, x}$ over a~$k$-letter alphabet~$\Sigma$ and let~$m_i $ be the width of~$\aut{A}$, for~$i\in[2d]$. Then, for~$i\in[d,2d]$ we have:
	\begin{align*}
		m_i =
		\begin{cases}
			k^{\lceil \frac{2d-i}{2} \rceil},   & \text{if } x=0; \\
			k^{\lfloor \frac{2d-i}{2} \rfloor}, & \text{if } x=1, \\
		\end{cases}
	\end{align*}
	and for~$i\in[d]$ we have~$m_i=m_{2d-1}$.
\end{lemma}
\begin{proof}
	Let us prove for $x=0$.
	\begin{enumerate}
		\item $(\forall i\in[d,2d]): m_i=k^{\lceil \frac{2d-i}{2} \rceil}$: \\
			Let~$w_1,w_2\in\Sigma^{2d-i}$ such that they differ at least in one even position. Now, let~$w_3 = \sigma^{2(i-d)}\R{w_1}$, for some~$\sigma\in\Sigma$. It is easy to see that~$w_1w_3\in L_{k, d, 0}$ but~$w_2w_3\notin L_{k, d, 0}$, so $w_1$ and $w_2$ have different quotients, and so they have to reach different states. Therefore, the number of states on rank $i$ of $\aut{A}$ is given by~$k^{\lceil \frac{2d-i}{2} \rceil}$, where the exponent is the number of odd integers between $i$ and~$2d-1$.
		\item $(\forall i\in[d]): m_i=m_{2d-i}$: \\
			Let us look at~$\R{\aut{A}}$, the~NFA for~$\R{L_{k,d}}$ given by reversing every transition in~$\aut{A}$ and swapping the initial with the final states. In fact, it is easy to see that~$L_{k,d,x}=\R{L_{k,d,x}}$, hence~$\lang(A)=\lang(\R{\aut{A}})$. In~1, we proved that~$m_j=k^{\lceil \frac{2d-j}{2} \rceil}$, for every rank~$j\in[d,2d]$. The~$i$-th rank in~$\aut{A}$ corresponds to the~$(2d-i)$-th rank in~$\R{\aut{A}}$, so that bound must be preserved.
	\end{enumerate}
	For~$x=1$, the number of states is~$k^{\lfloor \frac{2d-i}{2} \rfloor}$, where the exponent is the number of even integers between~$i$ and~$2d-1$, so the proof is similar.
%
%
\end{proof}

Then, we have the following result for the operational state complexity of intersection:
\begin{lemma}\label{lemma:intersection-necessary}
	Let $\aut{A}$ and~$\aut{B}$ be \dfas that accept $L_{k, d, 0}$ and $L_{k, d, 1}$ and $m_i$ and $n_i$ the widths of rank $i$ in $\aut{A}$ and $\aut{B}$, respectively, for $i\in [2d]$ and $d>0$. A \dfa that recognizes the language $L_{k, d, 0}\cap L_{k, d, 1}$ needs~$\sum_{i=0}^{2d}m_in_i+1$ states.
\end{lemma}
\begin{proof}

	As stated in~\cref{lemma:intersection-dfa-storage}, we have~$m_i=m_{2d-i}=k^{\lceil \frac{2d-i}{2} \rceil}$ and~$n_i=n_{2d-i}=k^{\lfloor \frac{2d-i}{2} \rfloor}$, for~$i\in[d, 2d]$. Moreover, it is easy to see that 
	$$L_{k,d,0}\cap L_{k,d,1} = \{ w\R{w} \mid w\in\Sigma^d\}\text,$$ 
	that is, the set of palindromes of even length. A minimal \dfa~$\aut{C}$ for this language with set of states $S = S_0\cup\ldots\cup S_{\ell}$ must first be able to remember the entire first half of the word, therefore, $|S_i|=k\cdot|S_{i+1}|=k^{2d-i}$ for $i\in[d, 2d-1]$. For the second half, it must check for the repetition of the first, then, $|S_i|=|S_{2d-i}|$, for $i\in[d]$. In fact,
	$$ |S_i| = m_in_i = k^{\lceil \frac{2d-i}{2} \rceil}k^{\lfloor \frac{2d-i}{2} \rfloor}=k^{2d-i}\text,$$	
	as desired.
\end{proof}

From~\cref{lemma:s-intersection-sufficient,lemma:intersection-necessary} we have:
\begin{theorem}\label{theorem:sc-intersection}
	Given two \dfas $\aut{A},\aut{B}$ for block languages $L_1,L_2\subseteq \Sigma^\ell$, for $\ell>0$, $\sum_{i=0}^{\ell}m_in_i +1$  states are necessary and sufficient in the worst-case for a \dfa that accepts the intersection of $L_1$ and $L_2$, where $m_i$ and $n_i$ are the widths of rank $i$ in $\aut{A}$ and $\aut{B}$, respectively, for $i\in [\ell]$.
	\end{theorem}
	
For the nondeterministic state complexity, the bounds are the same except that the dead state is not considered. In fact, the family witness languages  for the tightness of deterministic state complexity is also a witness for the nondeterministic one.
\begin{theorem}\label{theorem:nsc-intersection}
	Let~$\aut{A}=\langle Q,\Sigma,\delta_1,\state_0,\{q_f\}\rangle$ and~$\aut{B}=\langle P,\Sigma,\delta_2,p_0,\{p_f\}\rangle$ be  minimal \nfas for two block languages $L_1, L_2\subseteq\Sigma^\ell$, respectively, for some $\ell>0$, and such that $|Q| = m$ and $|P| = n$. Let $Q_i$ be the set of states in rank $i$ in $\aut{A}$ and $m_i=\width(i)=|Q_i|$, for $i\in [\ell]$. Let $P_i$ be the set of states in rank $i$ in $\aut{B}$ and $n_i=\width(i)=|P_i|$, for  $i\in [\ell]$. Additionally, $m=\sum_{i\in[\ell]}m_i$ and 
	$n=\sum_{i\in[\ell]}n_i$. 
	
Then, an \nfa with $\sum_{i=0}^{\ell}m_in_i$ states is sufficient to recognize the intersection of both languages and the bound is tight for $k>1$.
\end{theorem}
\begin{proof} 
	The fact that~$\sum_{i=0}^{\ell}m_in_i$ states are sufficient follows from the previous discussions.
	Moreover,
	this number of states is necessary,
	as can be noticed by considering the languages $L_{k, d, x}$ given above.
	Recall the language $L_{k, d, x}$, for some $k>1$, $d>0$ and $x\in\{0,1\}$.
	In fact,
	it is easy to see that $\dsc(L_{k, d, x})-1 =\nsc(L_{k, d, x})$,
	since the \nfa for $L_{k, d, x}$ must also be able to remember the same information as the \dfa.
	Then,
	if $\aut{A}$ (resp. $\aut{B}$) is a minimal \nfa that recognizes the language~$L_{k, d, 0}$
	(resp. $L_{k, d, 1}$),
	an \nfa that recognizes the intersection of both needs exactly~$\sum_{i=0}^{\ell}n_im_i$ states.
\end{proof}

\subsection{Union}
Let $L_1, L_2 \subseteq \Sigma^\ell$ be two block languages, for some $\ell>0$ and $|\Sigma| = k$, and their respective bitmaps $\bs(L_1), \bs(L_2)$. The bitmap of the union of~$L_1$ and~$L_2$ is~$\bs(L_1)\vee\bs(L_2)$.

Let~$\aut{A}=\langle Q\cup \{\Omega_1\},\Sigma,\delta_1,\state_0,\{q_f\}\rangle$ and~$\aut{B}=\langle P\cup \{\Omega_2\},\Sigma,\delta_2,p_0,\{p_f\}\rangle$ be the minimal \dfas for $L_1$ and $L_2$, respectively, with $|Q| = m$ and $|P| = n$. Again, let~$\aut{C}$ be the product \dfa of $\aut{A}$ and $\aut{B}$. Because $L_1,L_2$ are finite we know that $m+n$ states can be saved: $m+n-2$ because the initial states are non returning and $2$ more because the final states $(q_f, \Omega_2)$, $(\Omega_1, p_f)$, and $(q_f, p_f)$ can be merged into a single final state. However, again, one only needs to consider pairs of states $(q,p)$ such that $\rank(q) = \rank(p)$, for $q\in Q, p\in P$. Let  $Q_i$ be the set of states in rank $i$ in $\aut{A}$ and $m_i=\width(i)=|Q_i|$, for  $i\in [\ell]$. Let  $P_i$ be the set of states in rank $i$ in $\aut{B}$ and $n_i=\width(i)=|P_i|$, for  $i\in [\ell]$. Additionally, $m=1+\sum_{i\in[\ell]}m_i$ and $n=1+\sum_{i\in[\ell]}n_i$ since the dead states $\Omega_j$ do not belong to any rank. We have that

\begin{lemma}	\label{lemma:sc-union-sufficient}
Given two \dfas $A$ and $B$ for block languages $L_1$ and $L_2$, respectively, a \dfa with 
	$$\sum_{i=1}^{\ell-1}(m_in_i+m_i+n_i)+3$$ 
	states is sufficient to recognize the union of $L_1$ and $L_2$, where $m_i$ and $n_i$ are the widths of rank $i$ in $\aut{A}$ and $\aut{B}$, respectively, for $i\in [\ell]$.
\end{lemma}
\begin{proof}
	Let~$\aut{C}$ be the product automaton from $\aut{A}$ and~$\aut{B}$. As mentioned above, the final states $(q_f, \Omega_2)$, $(\Omega_1, p_f)$ can be merged with $(q_f, p_f)$, and a state $(p,q)$ is only accessible from the initial state if $\rank(p)=\rank(q)$. Therefore, the \dfa resulting from trimming $\aut{C}$ has a single initial state,  a final state and a dead state, and also the states $(Q_i\times P_i) \cup (Q_i\times\{\Omega_2\}) \cup (\{\Omega_1\}\times P_i)$, at each rank $i\in[1, \ell-1]$. Thus, the sufficient  number of states follows.
\end{proof}

In fact, the bound is tight for an alphabet with size at least $3$.

\begin{lemma}\label{lemma:sc-union-necessary}
Given two \dfas $A$ and $B$ for block languages $L_1$ and $L_2$ over $\Sigma^\ell$, respectively, a \dfa with $\sum_{i=1}^{\ell-1}(m_in_i+m_i+n_i)+3$ states is necessary to recognize the union of $L_1$ and $L_2$, where $m_i$ and $n_i$ are the widths of rank $i$ in $\aut{A}$ and $\aut{B}$, respectively, for $i\in [\ell]$ and $|\Sigma| > 2$.
\end{lemma}
\begin{proof}
	Since $\aut{A}$ and $\aut{B}$ are deterministic, $n_{\ell-1}$ and $m_{\ell-1}$, the number of states at rank $\ell-1$ of $\aut{A}$ and $\aut{B}$, respectively, are bounded by $k$ and not equal to $0$. Analogously, the width of the rank $\ell-1$ of the \dfa for the union of $\lang(\aut{A})$ and $\lang(\aut{B})$ is also at most $k$. When $k=2$, it is easy to see that the inequality $0< n_{\ell-1}m_{\ell-1} + n_{\ell-1} +m_{\ell-1} \leq k$ has no solutions.
	
	Now, consider the languages $L_{1,\ell} = \{a,c\}^\ell$ and $L_{2,\ell} = \{b,c\}^\ell$, and let $\aut{A}$ and $\aut{B}$ be the \dfas that recognize them, respectively, for some $\ell$ and $\Sigma=\{a,b,c\}$. We have that $\dsc(L_{1,\ell}) = \dsc(L_{2,\ell}) = \ell+2$ and $n_i = m_i = 1$, for every $i\in[\ell]$. The minimal \dfa that recognizes the language $L_{1,\ell}\cup L_{2,\ell}$ requires~$3$ states at each rank $i\in[1,\ell-1]$: one state  when some $a$ has already been read, so the word is in $L_{1,\ell}$; one state when some $b$ has already been read, so the word is in $L_{2,\ell}$; and one state for when only $c$'s have been read, so the \dfa still does not know to what particular language it belongs. Then, $\dsc(L_{1,\ell}\cup L_{2,\ell}) = \sum_{i=1}^{\ell-1}(n_im_i+n_i+m_i)+3 = 3\ell$.
\end{proof}

From \cref{lemma:sc-union-sufficient,lemma:sc-union-necessary} we have:

\begin{theorem}\label{theorem:sc-union}
		Given two \dfas $\aut{A},\aut{B}$ for block languages $L_1,L_2\subseteq \Sigma^\ell$, for $\ell >0$, $\sum_{i=1}^{\ell-1}(m_in_i+m_i+n_i)+3$  states are sufficient and necessary, if $\Sigma> 2$,   in the worst-case for a \dfa that accepts  the union  of $L_1$  and~$L_2$, where~$m_i$ and $n_i$ are the widths of rank $i$ in $\aut{A}$ and $\aut{B}$, respectively, for $i\in [\ell]$.
\end{theorem}

For the nondeterministic state complexity, the upper bound is the same as for finite languages over the same alphabet size.
\begin{theorem}	
	Let $L_1, L_2 \subseteq \Sigma^\ell$ with~$\ell>0$ and~$|\Sigma|=k$, such that $\nsc(L_1) = n$ and $\nsc(L_2) = m$. Then, $\nsc(L_1\cup L_2) \leq n+m-2$, and this bound is reached.
\end{theorem}
\begin{proof}
	Let $L_{1,\ell} = \{a^\ell\}$ and $L_{2,\ell} = \{b^\ell\}$, for some $\ell>0$ and $\Sigma=\{a,b\}$. We have that $\nsc(L_{1,\ell})=\nsc(L_{2,\ell})=\ell+1$ and~$\nsc(L_{1,\ell}\cup L_{2,\ell}) = 2\ell$.
\end{proof}

\subsection{Concatenation}
Consider two languages $L_1 \subseteq \Sigma^{\ell_1}$ and $L_2 \subseteq \Sigma^{\ell_2}$, for some $\ell_1, \ell_2 > 0$ and $|\Sigma|=k$, with bitmaps $\bs(L_1)$ and $\bs(L_2)$, respectively. The bitmap for the language $L_1L_2$ is given by replacing each $1$ in $\bs(L_1)$ by $\bs(L_2)$ and each~$0$ by $0^{k^{\ell_2}}$. This ensures that each word of $L_1L_2$ is obtained by concatenating a word of $L_1$ with a word of $L_2$ and for each word obtained in such way the correspondent bit in $\bs(L_1L_2)$ is set to $1$.

The deterministic state complexity of the concatenation for block languages coincides with the one for the finite languages when the first operand, $L_1$, has a single final state in its minimal \dfa. Therefore, we have the following exact upper bound:

\begin{theorem}
Let $L_1\subseteq \Sigma^{\ell_1}$ and $L_2\subseteq \Sigma^{\ell_2}$, for some $\ell_1, \ell_2 > 0$, be two block languages over a $k$-letter alphabet, where $\dsc(L_1)=m$ and $\dsc(L_2)=n$. Then, $\dsc(L_1L_2) = m + n - 2$.
\end{theorem}
\begin{proof}
	Let $\aut{A}$ and $\aut{B}$ be the minimal \dfas for $L_1$ and $L_2$, respectively. Also, let $\aut{C}$ be the minimal \dfa for~$L_1L_2$. Considering the bitmaps for these languages, the width of the rank $i$ of $\aut{C}$ is $|\bsr(L_2)|_i$, if $i\in[\ell_2]$, or is $|\bsr(L_1)|_{i-\ell_2}$, if $i\in[\ell_2+1,\ell_1+\ell_2]$. Then, $\aut{C}$ saves $2$ states by reusing the final state of $\aut{A}$ for the initial state of $\aut{B}$ (alternatively, reusing the initial state of $\aut{B}$ for the final state of $\aut{A}$) and also by eliminating one of the dead states.
\end{proof}

For the nondeterministic state complexity, the same result is expected, coinciding with the state complexity for the finite languages.
\begin{theorem}
Let $L_1\subseteq \Sigma^{\ell_1}$ and $L_2\subseteq \Sigma^{\ell_2}$, for some $\ell_1, \ell_2 > 0$, be two block languages over a $k$-letter alphabet, where $\nsc(L_1)=m$ and $\nsc(L_2)=n$. Then, $\nsc(L_1L_2) = m + n - 1$.	
\end{theorem}

In fact, any two languages $L_1 \subseteq \Sigma^{\ell_1}$ and $L_2 \subseteq \Sigma^{\ell_2}$ result in a family of witness languages. That is due to the fact that this operation preserves the ranks of the \dfas of the operands. 
\begin{example}\label{example:concatenation}
	Let $L_{1,\ell_1} = \{a^{\ell_1}\}$ and $L_{2,\ell_2} = \{a^{\ell_2}\}$, for $\ell_1, \ell_2>0$ and $\Sigma=\{a\}$. We have that $\dsc(L_{1,\ell_1})=\ell_1+2$, $\dsc(L_2)=\ell_2+2$, and $\dsc(L_{1,\ell_1}L_{2,\ell_2})=\ell_1+\ell_2+2$. We also have  $\nsc(L_{1,\ell_1})=\ell_1+1$, $\nsc(L_{2,\ell_2})=\ell_2+1$, and $\nsc(L_{1,\ell_1}L_{2,\ell_2})=\ell_1+\ell_2+1$.
\end{example}

\subsection{Block Complement}\label{section:block complement}
Consider a language $L \subseteq \Sigma^\ell$, for some $\ell>0$ and alphabet of size $k>0$, and let~$\bs$ be its bitmap. 
Now,
given a block language~$\Sigma^{\ell}$,
we consider block complement language,
namely $\Sigma^{\ell} \setminus L$, also 
denoted by $\overline{L}^\ell$.

Then, the bitmap of the language $\Sigma^\ell \setminus L$, namely~$\overline{\bs}$, is given by flipping every bit of~$\bs$.

\begin{theorem}\label{theorem:sc-blockcomplement}
	Let $L\subseteq \Sigma^\ell$, with $\ell>0$, be a block language with $|\Sigma|=k$, such that $\dsc(L) = m$. Then, $m - (\ell-1) \leq \dsc(\Sigma^\ell \setminus L) \leq m + (\ell-1)$.
\end{theorem}
\begin{proof}
	The number of states on a rank $i\in[\ell]$ of the minimal \dfa for $\Sigma^\ell \setminus L$ is given by the cardinality of~$\overline{\bsr}_i$, the set of the non-null factors of length $k^i$ on the bitmap~$\overline{\bs}$. If, for some $j\in[k^{\ell-i}-1]$, we have that $s^i_j = 0\cdots0$, which by definition implies that $s^i_j\notin\bsr_i$, then $\overline{s^i_j}=1\cdots1$ and so~$\overline{s^i_j}\in\overline{\bsr}_i$. Moreover, the complement may also occur. Therefore, $\big| |\bsr_i|-|\overline{\bsr}_i| \big| \leq 1$.
	
	Let $L_\ell=\{a^\ell\}$, for $\ell>0$. As we previously saw on \cref{theorem:witness-word-op}, $\dsc(L)=\ell+1$, and~$\dsc(\Sigma^\ell \setminus L) = 2\ell$.
\end{proof}
	
For the nondeterministic state complexity of the block complement operation, we have that the bound meets the one of the complement from the gmeneral case for finite languages considering the determinization cost of block languages. Also, this bound is asymptotically tight for alphabets of size at least $2$.

\begin{lemma}\label{lemma:nsc-blockcomplement-sufficient} 
	Let $L\subseteq \Sigma^\ell$ be a block language with $|\Sigma|=k$, such that $L$ is accepted by an $m$-state \nfa. Then, $2^{O(\sqrt{m})}$ states are sufficient for an \nfa for $\Sigma^\ell \setminus L$.
\end{lemma}
\begin{proof}
	Let $\aut{A}$ be a \nfa for $L$ with $m$ states. The minimal \dfa $\aut{B}$ for $L$ will have at most $2^{O(\sqrt{m})}$ states~\cite{KarhOkho:2O14}. Furthermore, the minimal \dfa $\aut{C}$ for $\Sigma^\ell \setminus L$ will have at most $\ell+1$ more states then $\aut{B}$, as shown in~\cref{theorem:sc-blockcomplement}. The nondeterministic state complexity is trivially bounded by the deterministic state complexity, so the sufficient number of states follows.
\end{proof}

Consider the following family presented by Karhumäki and Okhotin~\cite{KarhOkho:2O14}:
	$$ L_{k, d} = \cset{w_0\cdots w_{2d-1} \mid \exists i\in[d-1] : w_i = w_{i+d} \in \Sigma \setminus \{\sigma_{k-1}\} }$$ 
	defined over a $k$-ary alphabet $\Sigma = \{\letter_0, \ldots, \letter_{k-1} \}$. Informally, this language contains words that can be split into two halves of size $d$, such that there is at least one position in the first half that matches its counterpart in the second one, and it is different than the ``\emph{prohibited symbol}'' $\sigma_{k-1}$.

\begin{proposition}[\cite{KarhOkho:2O14}] \label{proposition:nsc-complement}
	For each $k\geq2$ and $d\geq2$, the language $L_{k,d}$ is recognized by an \nfa with $(k-1)d^2+2d$ states.
\end{proposition}

\begin{lemma} \label{lemma:nsc-complement}
	For each $k\geq2$ and $d\geq2$, the language~$\overline{L}_{k,d}^{2d}$ defined over a $k$-letter alphabet~$\Sigma$ requires at least $k^d$ states on the $d$-th rank.
\end{lemma}
\begin{proof}
	First, notice that~$\overline{L}_{k,d}^{2d}$ is the block complement of the language  $\overline{L}_{k,d}$ defined above, formally
	$$ \overline{L}_{k,d}^{2d} = \cset{ w_0\cdots w_{2d-1} \mid \forall i\in[d-1]: w_i \neq w_{i+d} \text{ or }  w_i = \sigma_{k-1} }\text.$$
	Let $w_1$ and $w_2$ be two words in $\Sigma^d$ such that~$a$ and~$b$ are the~$i$-th symbols of~$w_1$ and~$w_2$, respectively, with~$a\neq b$ and~$i\in[d]$. If~$a=\sigma_{k-1}$ then, with~$w_3=\sigma_{k-1}^{i-1}\,b\,\sigma_{k-1}^{d-i}$, we have~$w_1w_3\in\overline{L}_{k,d}^{2d}$ but~$w_2w_3\notin\overline{L}_{k,d}^{2d}$. If~$b=\sigma_{k-1}$ then, with~$w_3=\sigma_{k-1}^{i-1}\,a\,\sigma_{k-1}^{d-i}$, we have~$w_1w_3\notin\overline{L}_{k,d}^{2d}$ but~$w_2w_3\in\overline{L}_{k,d}^{2d}$. As a consequence,~$w_1^{-1}\overline{L}_{k,d}^{2d} \neq w_2^{-1}\overline{L}_{k,d}^{2d}$. Therefore, one state in rank~$d$ is needed for each word in~$\Sigma^d$.
	
	
\end{proof}

With these results, it is possible to determine that the nondeterministic state complexity for the complement operation given in \cref{lemma:nsc-blockcomplement-sufficient} is tight.

\begin{theorem}\label{theorem:nsc-complement}
	Let $m\geq2$ and $\Sigma$ an alphabet of size $k\geq2$. Then, there exists a language $L\subseteq\Sigma^\ell$, for some~$\ell>0$, such that $\nsc(L)=m$ and $\nsc(\Sigma^{\ell} \setminus L)=2^{\Omega\left(\sqrt{m}\right)}$. 
\end{theorem}
\begin{proof}
	Consider $d$ as the largest integer for which $(k-1)d^2+2d\leq m$. Following the work in~\cite{KarhOkho:2O14}, we have that
	$$d =  \left\lfloor \sqrt{\frac{m}{k-1} + \frac{1}{(k-1)^2}} - \frac{1}{k-1} \right\rfloor\text.$$
	Then, $L_{k,d}$ is a language recognized by an $m$-state \nfa, while every \nfa for $\overline{L}_{k,d}^\ell$ requires, by Lemma~\ref{lemma:nsc-complement}, at least
	$$k^d = k^{\left\lfloor \sqrt{\frac{m}{k-1} + \frac{1}{(k-1)^2}} - \frac{1}{k-1} \right\rfloor} \geq k^{\sqrt{\frac{m}{k-1}}-2} =2^{\Omega(\sqrt{m})}$$
	states, as required.
\end{proof}

\subsection{Kleene Star and Plus}\label{section:starandplus}
Let $L\subseteq\Sigma^\ell$, for some $\ell>0$ and $\bs$ its bitmap. From $\bs$ one can obtain the minimal \dfa for $L$, namely $\aut{A}=\langle Q,\Sigma,\delta_0,\state_0,\{\state_f\}\rangle$.

 A \dfa $\aut{B}=\langle Q\setminus \{ q_f\},\Sigma,\delta_1,\state_0,\{\state_0\}\rangle$ recognizes the language $L^\star$ if $\delta_1(\state,\letter)=q_0$, for all $\state\in Q$ such that $\rank(\state)=1$ and $\letter\in\Sigma$, and $\delta_1(\state,\letter)=\delta(\state,\letter)$, for the remaining pairs $(\state, \letter)\in Q\times\Sigma$. That is, the \dfa for $L^\star$ is given by substituting  all the transitions with final state as the target state to transitions to the initial state. The same applies for the \nfa for $L^\star$, as the following theorem states.

\begin{theorem}
	Let $L\subseteq \Sigma^\ell$, with $\ell>0$, be a block language with $\dsc(L) = n$ and $\nsc(L) = m$. Then, $\dsc(L^\star) = n-1$ and $\nsc(L^\star) = m-1$.
\end{theorem}

Moreover, a \dfa $\aut{C}=\langle Q,\Sigma,\delta_2,\state_0,\{\state_f\}\rangle$ recognizes the language $L^+$ if $\delta_2(q_f, \letter) = \delta_0(q_0, \letter)$, for $\letter\in\Sigma$. Again, the \nfa for $L^+$ is given by applying the same changes to the minimal \nfa for $L$. And the witness languages coincide with the ones for finite languages, namely 
$L_{\ell} = \{a^{\ell}\}$ , for $\ell> 0$.
\begin{theorem}
	Let $L\subseteq \Sigma^\ell$, with $\ell>0$. Then, $\dsc(L^+) = \dsc(L)$ and $\nsc(L^+) =\nsc(L)$.
\end{theorem}

\section{Conclusions} \label{section:conclusions}
The complexities obtained for operations on block languages are summarized in \cref{tab:cblock}. One can compare these results with the ones for finite languages summarized in \cref{tab:cfin}. For the deterministic state complexity, the bounds for Boolean operations on block languages are given using the rank widths and are smaller than the ones for finite languages. It would be interesting to express them as a function of the number of states of the operands (as it is usually done). Moreover, those bounds could be obtained from a direct construction of the minimal \dfa for the resulting language considering the bitmaps of the operands. Note that bitwise Boolean operations can be performed to obtain the bitmap factors (i.e., states) in each rank of the resulting \dfa. This study will be interesting to pursue in future work. For concatenation and Kleene star the bounds correspond to special cases of the ones for finite languages. Finally, for reversal the results are analogous to the ones for finite languages, but here considering the bounds known for the determinization of block languages. The results for nondeterministic state complexity meet the values known for finite languages except for intersection and the specific operations for block languages (block complement, word addition, and word removal).

\begin{table}
	\centering
	\caption{Upper bounds of the state complexity for block languages of words of length $\ell$.} \label{tab:cblock}
	\begin{tabular}{lcccc}
		\toprule
		& \multicolumn{2}{c}{Block Languages} \\
		\midrule
		& \multicolumn{1}{c}{sc} & \multicolumn{1}{c}{$|\Sigma|$} & \multicolumn{1}{c}{nsc} & \multicolumn{1}{c}{$|\Sigma|$} \\
		\midrule
		$L_1\cup L_2$ & $\sum_{i=1}^{\ell-1}(m_in_i+m_i+n_i)+3$ & $3$ & $m+n-2$ & $2$ \\
		$L_1\cap L_2$ & $\sum_{i=0}^{\ell}m_in_i+1$ & $2$ & $\sum_{i=0}^{\ell}m_in_i$ & $2$ \\
		$L_1L_2$ & $m+n-2$ & $1$ & $m+n-1$ & $1$ \\
		$\Sigma^\ell\setminus L$ & $m+\ell-1$ & $2$ & $O(2^{\sqrt{m}})$ & $2$ \\
		$L\cup\{w\}$ & $m+\ell-1$ & 2 & $m+\ell-1$ & 2 \\
		$L\setminus\{w\}$ & $m+\ell-1$ & 2 & $m+\ell-1$ & 2 \\
		$L^*$ & $m-1$ & $1$ & $m-1$ & $1$ \\
		$L^+$ & $m$ & $1$ & $m$& $1$ \\
		$\R{L}$ & $2^{\Theta(\sqrt{m})}$ & $2$ & $m$ & $1$ \\
		\bottomrule  
	\end{tabular}
\end{table}

\bibliographystyle{eptcs}
\bibliography{flan16}
\newpage
\appendix

\end{document}

%% file: TiKZit-General.tikzstyles
\tikzstyle{black node}=[fill=black, draw=black, shape=circle, inner sep=.5]
\tikzstyle{simple node}=[fill={rgb,255: red,226; green,226; blue,226}, draw=black, shape=circle]
\tikzstyle{simple node 2}=[fill={rgb,255: red,226; green,226; blue,226}, draw=black, shape=circle, inner sep=.5]
\tikzstyle{vertex}=[fill=black, draw=black, shape=circle, tikzit draw=red]
\tikzstyle{graph node}=[fill=black, draw=black, shape=circle, inner sep=4]

\tikzstyle{thick line 0}=[-, very thick]
\tikzstyle{thin line 1}=[-, very thin, line width=0.5]
\tikzstyle{dashed line thin}=[-, very thin, dashed]
\tikzstyle{dotted line thin}=[-, very thin, densely dotted, line width=0.5]
\tikzstyle{very thin solid}=[{|-|}, very thin, solid]
\tikzstyle{filled light}=[-, fill={rgb,255: red,199; green,199; blue,199}, tikzit fill={rgb,255: red,217; green,217; blue,217}]
\tikzstyle{filled darker}=[-, fill={rgb,255: red,135; green,135; blue,135}, tikzit fill={rgb,255: red,135; green,135; blue,135}]
\tikzstyle{arrow right}=[->, line width=0.5]

%% file: macros.tex
\DeclareMathOperator{\REV}{R}
\DeclareMathOperator{\AD}{D}

\DeclareMathOperator{\dsc}{\mathsf{sc}}
\DeclareMathOperator{\nsc}{\mathsf{nsc}}

\newcommand{\letter}{\sigma}

\newcommand{\dfa}{DFA\xspace}
\newcommand{\dfas}{DFAs\xspace}

\newcommand{\nfa}{NFA\xspace}

\newcommand{\nfas}{NFAs\xspace}

\newcommand{\state}{q}

\newcommand{\lang}{\mathcal{L}}

\newcommand{\R}[1]{#1^{\REV}}

\newcommand{\cset}[1]{\{\,#1\,\}}

\newcommand{\bsset}[1]{{\ensuremath\mathcal{B}_{#1}}}

\DeclareMathOperator{\bs}{{\mathsf{B}}}
\DeclareMathOperator{\bsr}{{\mathcal{B}}}
\DeclareMathOperator{\cover}{{\mathcal{C}}}
\DeclareMathOperator{\width}{\mathsf{w}}
\DeclareMathOperator{\rank}{\mathsf{rank}}
\DeclareMathOperator{\indi}{\mathsf{ind}}
\DeclareMathOperator{\pad}{\mathsf{pad}}
\DeclareMathOperator{\maxmindfa}{{\mathsf{MAX}}}
\newcommand\aut[1]{#1}

\newcommand{\m}{t_\ell}
\newcommand{\maxr}{r_\ell}
\newcommand{\bin}[1]{{#1}_{[2]}}

\newcommand{\qee}{\makeatletter
\def\@eqnnum{{\normalfont \normalcolor \hbox{\qed}}}
\makeatother
}